\documentclass[english]{article}

\usepackage{custom_tex}

\author{Sifan Liu\thanks{Correspondence email: \texttt{sfliu@stanford.edu}}}
\author{Jelena Markovic-Voronov}
\author{Jonathan Taylor}
\affil{Department of Statistics, Stanford University}
\date{August 2023}
\title{Black-box Selective Inference via Bootstrapping}

\usepackage{amsmath}
\usepackage{graphicx}
\usepackage{enumerate}
\usepackage{url} 

\newcommand{\rd}{\mathrm{d}}
\newcommand{\whH}{\widehat{H}}
\newcommand{\whM}{\widehat{M}}
\newcommand{\whU}{\widehat{U}}
\newcommand{\whV}{\widehat{V}}
\newcommand{\whW}{\widehat{W}}
\newcommand{\whZ}{\widehat{Z}}
\newcommand{\bzero}{\mathbf{0}}

\newtheorem{assumption}[theorem]{Assumption}

\begin{document}
\maketitle

\begin{abstract}
Conditional selective inference requires an exact characterization of the selection event, which is often unavailable except for a few examples like the lasso. This work addresses this challenge by introducing a generic approach to estimate the selection event, facilitating feasible inference conditioned on the selection event. The method proceeds by repeatedly generating bootstrap data and running the selection algorithm on the new datasets. Using the outputs of the selection algorithm, we can estimate the selection probability as a function of certain summary statistics. This leads to an estimate of the distribution of the data conditioned on the selection event, which forms the basis for conditional selective inference. We provide a theoretical guarantee assuming both asymptotic normality of relevant statistics and accurate estimation of the selection probability. The applicability of the proposed method is demonstrated through a variety of problems that lack exact characterizations of selection, where conditional selective inference was previously infeasible.
\end{abstract}

\section{Introduction}
\label{sec:intro}

Inference after selection is susceptible to bias when the same information is employed both for the selection process and the subsequent inference. This double use of data can lead to skewed outcomes. To conduct valid inference, the information that has been used for selection must be discarded. The conditional selective inference framework addresses this issue by conditioning on the selection event and conducting inference based on the conditional distribution of the data.

In certain cases, the conditional distribution given the selection event is tractable. For instance, in the context of the lasso \citep{tibshirani1996regression} selection, \citet{lee2016exact} have derived the exact distribution of the data conditioned on the signs of the lasso solution under Gaussian noise. This methodology is further extended to the square-root lasso \citep{tian2018squareroot}, forward stepwise regression, and least-angle regression \citep{taylor2014post}. To increase inferential power, randomized versions of the lasso have been proposed, including using a subset for selection and adding noise to the data \citep{tian2018selective}. In the case of the randomized lasso, specific algorithms have been developed to facilitate valid and efficient inference \citep{panigrahi2022approximate,panigrahi2022exact}. Nonetheless, these methods are tailored exclusively to selection algorithms like the lasso and its variants. For more complex selection procedures, the distribution of the data conditioned on the selection event often lacks an exact characterization, thereby limiting the feasibility of this conditional approach.

The objective of this work is to address this challenge by introducing a generic approach to estimate the conditional distribution and thus enable feasible inference even in scenarios where exact characterization of the conditional distribution is elusive. To provide an initial glimpse into the approach, suppose the selection event is denoted as $\{\whM=M \}$, and the aim is to conduct inference for the parameter $\theta$ based on the conditional distribution of the test statistic $\hat\theta$ given the selection. We start by assuming that there exists a statistic $\whU$ that is independent of $\hat\theta$ and, importantly, the selection process depends on the data only through $(\hat\theta,\whU)$. This allows the factorization of the conditional density of $\hat\theta\mid\{\whM=M, \whU=U \}$ as follows:
\begin{align*}
p_{\hat\theta}(x\mid \whM=M, \whU=U;\theta )\propto p_{\hat\theta}(x;\theta)\cdot \PP{\whM=M\mid \hat\theta=x, \whU=U }.
\end{align*}
The two terms on the right-hand side are the pre-selection density of $\hat\theta$ and the probability of selecting the model $M$ given $\hat\theta$ and $\whU$. 
Because $(\hat\theta,\whU)$ is assumed to be sufficient for the model selection, the probability $\PP{\whM=M\mid \hat\theta=x, \whU=U }$ does not depend on the data and is solely a property of the selection algorithm itself. The core idea is to acquire knowledge about the selection probability by repeatedly executing the selection algorithm on newly created data.

Specifically, we repeatedly generate new datasets by bootstrapping and run the selection algorithm on the newly generated datasets. During this process, we keep track of the labels that indicate whether the selected model is the same as $M$, along with the summary statistics that are assumed to be sufficient for the selection process. These binary labels and the summary statistics are then employed to estimate the selection probability by minimizing the cross-entropy loss within a function class like a neural network. As long as the function class is representative enough and we generate a sufficiently large dataset that contains both types of labels, this approach is expected to yield a reliable estimate of the selection probability.

The contributions and structure of this paper are summarized as follows. In Section~\ref{sec: problem}, we provide some background on selective inference and introduce the proposed method. We also discuss the possibility to condition less by marginalizing over some ancillary statistic. We demonstrate the ideas using a running example of the drop-the-loser design to help understanding. In Section~\ref{sec: learning}, we delve into the implementation details. We explain how to generate bootstrap data, estimate the selection probability, and conduct inference once the selection probability is obtained. Additionally, we present a practical method for assessing the accuracy of the estimated conditional distribution. Section~\ref{sec: theory} establishes the asymptotic coverage guarantee contingent upon the selection probability being estimated sufficiently accurate and all summary statistics satisfying asymptotic normality under the pre-selection and bootstrapping distributions. In Section~\ref{sec: simu}, we apply the proposed method to a variety of problems. First, we compare the proposed method within the lasso problem against prior, more specialized methods. Then we consider scenarios where conditional selective inference was previously infeasible. This encompasses the tasks of conducting inference for parameters selected by some screening procedures, including the Benjamini-Hochberg procedure and the knockoff filter. Furthermore, we apply our method to a sequential testing setting where repeated tests are performed until achieving significance. Overall, our method is shown to yield valid statistical inference for these tasks, thus offering a solution to conduct conditional selective inference for a much broader spectrum of problems.

\section{Problem formulation}
\label{sec: problem}

Consider scenarios where data analysis is carried out, revealing underlying patterns within the data, generating initial hypotheses and conjectures, and influencing subsequent research and experimental decisions. Several illustrative examples come to mind:
\begin{itemize}
  \item \textbf{Initial screening of significant features or hypotheses:} Imagine situations like variable selection in regression or multiple hypotheses testing. Such procedures yield a model, which in this context represents a subset of potentially important variables or significant hypotheses. Following the initial screening, the next step is to conduct inference for the selected variables or hypotheses.
  \item \textbf{Selective reporting:} Consider multiple research laboratories independently conducting hypothesis testing on separate datasets. However, they only share their outcomes if certain hypothesis test is statistically significant. This phenomenon often results in ``publication bias". More generally, similar types of bias can emerge due to selective reporting. In this situation, the model is the indicator of whether the result is reported and the objective is to conduct valid inference with accessibility only to the reported data.
  \item \textbf{Adaptive clinical trial and sequential decision making:} Adaptive trial designs permit modifications to the trials based on the data observed thus far. In a two-stage design, for instance, findings from the first stage can guide decisions for the second stage. In particular, a subset of treatments or a subpopulation might be selected to continue into the second stage based on the preliminary findings. Here, the model is the decision made based on historical data, such as the selected subset of treatments or subpopulation. The objective is to leverage all collected data for inference regarding parameters of interest, such as the treatment effects. A similar example is known as repeated significance testing, where a hypothesis is tested repeatedly while accumulating new data until achieving significance. 
\end{itemize}
Across all these scenarios, the model $M$ depends on the data. Neglecting this dependence inadvertently introduces bias into inference. A common remedy for this bias is to condition on the selection event. By doing so, the inference does not use the information that has been used for selection thus avoids the bias. In the framework of conditional inference, it is crucial to have a characterization of the selection event. However, situations where the selection has a closed form are very rare. 

In this paper, we propose a method to estimate the selection event that is applicable to a much broader range of problems. In the below, let $D$ denote the dataset which follows an unknown distribution $\bbF$. Let $M$ denote the selected model. Let $\theta=\theta(\bbF,M)\in\R^s$ denote the parameter of interest and let $\hat\theta=\hat\theta(D,M)\in\R^s$ be the statistic whose distribution will be used to conduct inference for $\theta$.

\subsection{The proposed method}
Central to our approach is the assumption that there exists certain $d$-dimensional summary statistic $\whZ=\whZ(D)\in\R^d$ that is sufficient for the selection algorithm. In other words, we assume the model $\whM$ depends on the data $D$ solely through $\whZ$. If the selection algorithm is random, such as the lasso with a random response \citep{tian2018selective}, then $\whM$ also depends on some external noise denoted as $\omega$. We represent the selection procedure as the function $\calS$ such that $\whM=\calS(\whZ,\omega)$. We refer to $\whZ$ as the \emph{basis}.

Moreover, the basis $\whZ$ is assumed to be approximately jointly Gaussian with $\hat\theta$. In this section, we assume they follow the exact normal distribution
\begin{align}
\begin{pmatrix}
\hat\theta \\ \whZ
\end{pmatrix} \sim \N_{s+d}\left(\begin{pmatrix} \theta \\ \mu_Z \end{pmatrix},
\begin{pmatrix}
\Sigma & C\tran \\ C & \Sigma_Z
\end{pmatrix}
\right).
\label{equ: theta Z joint normal}
\end{align}
Consequently, we can decompose $\whZ$ as
\[
\whZ=\Gamma \hat\theta + \whU,\quad \text{where }\Gamma=C\Sigma^{-1}.
\]
This ensures that $\hat\theta$ is independent of $\whU$, denoted as $\hat\theta\indep\whU$, under the pre-selection distribution $\bbF$.
As a result, the density of $\hat\theta$ conditional on $\{\widehat M=M, \widehat U=U \}$ is proportional to
\begin{align*}
  p_{\hat\theta}(x \mid \widehat M=M, \widehat U=U; \theta )\propto  \varphi(x;\theta,\Sigma) \cdot \PP{\widehat M=M\mid \widehat U=U, \hat\theta=x},
\end{align*}
where $\varphi(x;\theta,\Sigma)$ is the probability density function (pdf) of the normal distribution $\N(\theta,\Sigma)$.
We define the selection probability function $\pi(Z)=\PP{\calS(\whZ, \omega )=M\mid \whZ=Z }$. Because $\whZ=\Gamma \hat\theta + \whU$, along with $\omega$, completely determines the selection procedure, it follows that $\PP{\widehat M=M\mid \widehat U=U, \hat\theta=x}=\pi(\Gamma x+U)$. Hence, the conditional density of $\hat\theta\mid \{\widehat M=M, \widehat U=U \}$ can be expressed as
\begin{align*}
  p_{\hat\theta}(x\mid \widehat M=M, \widehat U=U; \theta )\propto  \varphi(x;\theta,\Sigma) \cdot \pi(\Gamma x+U)  .
\end{align*}
While certain selection algorithms, like the lasso, possess a closed-form selection probability $\pi(\cdot)$, more complicated selection procedures lack a direct expression. Therefore, we propose to acquire knowledge about $\pi(\cdot)$ by executing the selection algorithm repeatedly. Before we introduce the method, we first present some intuitions using an illustrative example.

\subsection{A first example --- drop-the-losers design}
\label{sec: dtl}
Adaptive designs are frequently employed to expedite clinical trial durations. We consider the two-stage drop-the-losers (DTL) design, where the superior treatment from the first stage continues to the second stage \citep{sampson2005drop}. Suppose there are $K$ distinct treatments. In the first stage, we observe $X_{k,j}\iid\N(\theta_k,1)$ independently for $1\leq k\leq K$ and $1\leq j\leq n_1$, and compute the mean responses $\widebar {X}_k=\frac{1}{n_1}\sum_{j=1}^{n_1} X_{k,j} $. Subsequently, we select treatment $k^*=\argmax_{k\in[K]}\widebar X_k $, the one with the highest mean effect according to the first stage experiment, and administer $n_2$ additional subjects for treatment $k^*$, denoted by $Y_{k^*,j}\iid\N(\theta_{k^*},1)$ ($1\leq j\leq n_2$). The objective is to test $H_0:\theta_{k^*}=0 $ or construct a confidence interval for $\theta_{k^*}$ with the two stages of data. The conventional z-test employing the statistic $\hat\theta:=\frac{n_1\widebar X_{k^*} + n_2 \widebar Y  }{n_1+n_2} $ is biased because the true distribution of $\hat\theta$ is stochastically greater than $\N(\theta_{k^*},\frac{1}{n_1+n_2} )$ due to the selection in the first stage. A valid approach is to only use the second-stage data for inference, relying on the distribution of $\widebar Y_{k^*}\sim\N(\theta_{k^*},\frac{1}{n_2})$. However, this ``sample splitting" approach does not use any data from the first stage for inference. A more efficient strategy is to employ both stages of data while conditioning on the the selection of treatment $k^*$.

In this example, we know the selection is based on the summary statistic $\whZ=(\widebar X_{1},\ldots,\widebar X_K)$. Moreover, we know $(\hat\theta,\whZ)$ follows the normal distribution in Equation~\eqref{equ: theta Z joint normal} with
\begin{align*}
\mu_Z=(\theta_1,\ldots,\theta_K )\tran,\; \Sigma=\frac{1}{n_1+n_2},\; C=\frac{1}{n_1+n_2} e_{k^*}, \; \Sigma_Z=\frac{1}{n_1} I_K, 
\end{align*}
where $e_{k^*}$ is the $k^*$-th standard basis vector in $\R^K$. Since the parameter $\theta_{k^*}$ is of $s=1$ dimension, we also denote $\sigma^2=\Sigma=\frac{1}{n_1+n_2}$.
In this case, the decomposition of $\whZ$ can be expressed as $\whZ=e_{k^*} \hat\theta + \whU$, where
\[
\whU_j=\begin{cases}
\widebar X_j - \hat\theta &\text{ if }j=k^*,\\
\widebar X_j & \text{ otherwise }.
\end{cases}
\]
Because the selection event is $\left\{k^*=\argmax_{k\in[K]} Z_k \right\}$, we have {\sloppy$\pi(Z)=\Indc{Z_{k^*}\geq Z_j,\, \forall\, j\neq k^*}$}. Therefore, $\pi(\Gamma x + U)=\Indc{x + U_{k^*}\geq U_{j},\, \forall j\neq k^* } $.
Let $a=\max_{j\neq k^*} U_j$ and $b=U_{k^*}$. Then the conditional density of $\hat\theta\mid \whM=M, \whU=U $ can be expressed as
\begin{align}
\label{equ: dtl likelihood}
p_{\hat\theta}(x\mid \whM=M, \whU=U; \theta_{k^*} )\propto \varphi(x; \theta_{k^*},\sigma^2)  \cdot \Indc{x\geq a-b},\,\text{ where }\, a-b=\max_{j\neq k^*}U_j - U_{k^*}.
\end{align}
This indicates that the conditional distribution is the normal distribution $\N(\theta_{k^*}, \sigma^2)$ truncated to the interval $[a-b,\infty)$. We denote this truncated normal distribution as
\begin{align}
\label{equ: dtl trunc normal}
\text{TN}(\theta_{k^*}, \sigma^2; a-b, +\infty ).
\end{align}
To test $H_0:\theta_{k^*}=0$ versus the one-sided alternative $H_1:\theta_{k^*}>0$, the p-value can be computed as the tail probability
\[
1 - \frac{\Phi(\frac{\hat\theta}{\sigma})- \Phi(\frac{a-b}{\sigma})  }{1 - \Phi(\frac{a-b}{\sigma}) },
\]
where $\Phi$ is the cumulative distribution function (CDF) of the standard normal distribution.
Similarly, a level-($1-\alpha$) confidence interval for $\theta_{k^*}$ can be constructed by inverting the test:
\begin{align*}
  \text{CI}(\theta_{k^*})=\left\{\theta: 1 - \frac{\Phi(\frac{\hat\theta-\theta}{\sigma})- \Phi(\frac{a-b-\theta}{\sigma})  }{1 - \Phi(\frac{a-b-\theta}{\sigma}) }\geq \alpha  \right\}.
\end{align*}

\subsection{Condition less}
\label{sec: condition less}

Readers might have noted that our approach conditions on $\{\whM=M, \whU=U \}$, which encompasses more information than the necessary conditioning event $\{\whM=M \}$. Conditioning more indicates that we retain less information for inference, potentially resulting in lower power. In this section, we discuss how to condition less by marginalizing over some ancillary statistic.

Suppose it is known that the statistic $\whV\in\R^d$ is independent of $\hat\theta$ and follows the normal distribution $\whV\sim\N(0,\Sigma_V)$, then we can decompose $\whU$ into a part that is independent of $\whV$ and a part that is dependent on $\whV$. Without loss of generality, we write this decomposition as
\[
\whU = \whV + \whW,\quad \whV\indep \whW.
\]
Otherwise, we redefine $\whV$ as $\Cov{\whU, \whV} \Var{\whV}^{-1}\whV $.
The conditional density of $\hat\theta\mid \{\whM=M,\whW=W \}$ is then proportional to
\begin{align*}
  p_{\hat\theta}(x\mid \whM=M, \whW=W; \theta)&\propto \varphi(x; \theta, \Sigma) \cdot \PP{\whM=M\mid \hat\theta=x, \whW=W }\\
  &\propto \varphi(x; \theta, \Sigma) \cdot \int  \PP{\whM=M\mid \hat\theta=x, \whW=W, \whV=v } \varphi(v;0,\Sigma_V)\rd v \\
  &\propto \varphi(x; \theta, \Sigma) \cdot \int \pi(\Gamma x + v + W )\varphi(v;0,\Sigma_V)\rd v,
\end{align*}
where the last equality is due to $\whZ=\Gamma \hat\theta+\whV+\whW$.
Let $\tZ=\Gamma\hat\theta + \whW$ so that $\tZ=\whZ-\whV$ and $\tZ\indep \whV$. 
Define
\begin{align}\label{equ: define pi tilde}
\tilde\pi(\tZ):= \int \pi(\tZ + v )\varphi(v;0,\Sigma_V)\rd v= \PP{\calS(\tZ+\whV, \omega )=M \mid \tZ }.
\end{align}
Hence, the conditional density of $\hat\theta\mid \{\whM=M, \whW = W \} $ can be expressed as
\begin{align}
  \label{equ: cond density tilde pi}
  p_{\hat\theta}(x \mid \whM=M, \whW = W; \theta ) \propto \varphi(x; \theta,\Sigma)\cdot \tilde{\pi}(\Gamma x + W ).
\end{align}
Therefore, when conditioned on $\{\whM=M, \whW=W \}$, we need to estimate the selection probability $\tilde \pi$ as a function of $\tZ$. If $\whV$ is not a constant (i.e. 0), the event $\{\whM=M, \whW = W \}$ contains strictly less information than $\{\whM=M, \whU = U \}$, potentially preserving more information for the inference stage. In the rest of the paper, our focus is on estimating the conditional density in Equation~\eqref{equ: cond density tilde pi}, since $\whV$ can be trivially defined to be the constant 0.

\paragraph*{Revisit the DTL example}

In Section~\ref{sec: dtl}, we have described our procedure for the drop-the-loser (DTL) problem. Notice that
\[
\whU_{k^*}=\widebar X_{k^*} - \hat\theta \sim\N(0, s^2 ),
\]
where $s^2=\frac{1}{n_1}-\frac{1}{n_1+n_2}$. We let $\whV=e_{k^*}\whU_{k^*}$ and $\whW=\whU-\whV$. This choice of $\whV$ ensures that $\whV\indep \whW$. We also have $\tZ=\whZ-\whV$ with $\tZ_{k^*}=\hat\theta$ and $\tZ_j=\widebar X_j$ for $j\neq k^*$. As a result,
\begin{align*}
  \tilde\pi(\tZ)&=\int \Indc{ \hat\theta + v\geq a }\varphi(v;0,s^2)\rd v=\Phi(\frac{\hat\theta-a}{s} ),
\end{align*}
where $a=\max_{j\neq k^*}W_j=\max_{j\neq k^*}U_j$ as before. Compared to $\pi(Z)=\Indc{Z_{k^*}\geq a-b}$, which involves the hard truncation, the function $\tilde\pi$ is smooth due to the marginalization over the variable $\whU_{k^*}$. When conditioning on $\{\whM=M, \whZ=Z \}$, we condition on all the first-stage group means $\widebar X_j$ ($1\leq j\leq K$). In contrast, when conditioning on $\{\whM=M, \whW=W\}$, we only condition on the group means $\widebar X_j$ for all $j\neq k^*$, along with the event of selecting $k^*$ in the first stage.

With the expression of $\tilde\pi$, the conditional density of $\hat\theta\mid \{\whM=M,\whW=W\}$ is given by
\begin{align}
\label{equ: dtl density marginalize}
p_{\hat\theta}(x;\whM=M, \whW=W; \theta_{k^*})= \frac{\varphi(x;\theta_{k^*},\sigma^2) \Phi(\frac{x-a}{s})} {1 - \Phi(\frac{a-\theta_{k^*}} {\sqrt{\sigma^2+s^2}}) },
\end{align}
where $\sigma^2=\frac{1}{n_1+n_2}$. 

~\citet{kivaranovic2020tight} studied the length of the confidence intervals based on the conditional density~\eqref{equ: dtl density marginalize}. According to~\citet[Theorem 1]{kivaranovic2020tight}, the expected length of the one-sided level-$(1-\alpha)$ confidence interval is smaller than $\frac{\sigma\sqrt{\sigma^2+s^2} }{s}\Phi^{-1}(1-\alpha)=\frac{1}{\sqrt{n_2}} \Phi^{-1}(1-\alpha)$, implying that the confidence interval is on average shorter than the interval obtained by only using the second stage of data. Moreover, the confidence interval based on the hard truncated normal distribution~\eqref{equ: dtl trunc normal} has infinite expected length~\citep{kivaranovic2021length}, suggesting that marginalization (when feasible) leads to more powerful inference.

\section{Estimating the selection probability}
\label{sec: learning}
As discussed earlier, a key aspect of our method is to estimate the selection probability $\tilde\pi(\tZ)$ defined in Equation~\eqref{equ: define pi tilde} as a function of $\tZ$. This section outlines the estimation process. The primary assumption is the ability to repeatedly generate new datasets $D^*$ and run the selection algorithm on them, often facilitated by computational tools. Consequently, subjective model selection made by researchers falls beyond the scope of the proposed method.

\subsection{Generate training data}
\label{sec: generate training}
To initiate the process, we propose to generate datasets $D_1^*,\ldots,D_B^*$ through bootstrapping from the original dataset $D$. We denote the bootstrap distribution as $\bbF^*$ conditional on the dataset $D$. For each $D_i^*$, the selection algorithm is executed to obtain the selected model $M_i^*$. We record the labels $\ell_i^*=\Indc{M_i^*=M}$ and compute the basis $\tZ_i^*=\tZ(D_i^*)$. The choice of an appropriate basis $\tZ$ depends on the specific problem, and we will delve into the specifics of making these choices for the examples in Section~\ref{sec: simu}.

The set of data points $\{(\tZ_i^*,\ell_i^*)\}_{1\leq i\leq B}$ constitutes our training data. Based on the generating process, they follow the distribution
\[
\ell_i^*\sim \text{Bernoulli}(\tilde\pi^*(\tZ_i^* ) ),
\]
where $\tilde\pi^*(\tZ^*):= \bbF^*\{\whM^*=M\mid \tZ^* \}$ is the selection probability given $\tZ^*$ under the bootstrap distribution $\bbF^*$. So from the training data, what we can hope to estimate is, in fact, the function $\tilde\pi^*$, not $\tilde\pi$. However, the two functions can be closely related under bootstrap consistency.

To see this, assume that under the bootstrap distribution, suitably scaled, we have the approximate normal distribution
\begin{align*}
  \begin{pmatrix} \tZ^*-\tZ \\ \whV^* - \whV \end{pmatrix} \dot{\sim}\, \N\left(\begin{pmatrix}\mathbf{0} \\ \mathbf{0}\end{pmatrix}, \begin{pmatrix}
    \Sigma_{\tZ} & \bzero \\ \bzero & \Sigma_V
  \end{pmatrix} \right),
\end{align*}
where $\whV^*=\whZ^*-\tZ^*$.
Then we arrive the approximation
\begin{align*}
  \tilde\pi^*(\tZ^*)&=\int \bbF^*\{\whM^*=M \mid \tZ^*, \whV^*=v \}\cdot p_{\whV^*\mid \tZ^*}(v)\rd v\\
  &=\int \bbF^*\{\calS(\tZ^*+v )=M \mid \tZ^*, \whV^*=v \}\cdot p_{\whV^*\mid \tZ^*}(v)\rd v\\
  &\approx \int \PP{\calS(\tZ^*+v )=M \mid \tZ^*}\cdot  \varphi(v;\whV, \Sigma_V)\rd v\\
  &=\tilde\pi(\tZ^*+\whV),\numberthis\label{equ: pi star and pi}
\end{align*}
where the second equality is by the assumption that selected model $\whM^*$ only depends on $\whZ^*$ which is equal to $\tZ^*+\whV^*$, the third line is deduced from the above approximate normal distribution, and the last equality is by definition of $\tilde\pi$.
The approximate normal distribution will be justified by asymptotic normality of bootstrap in Section~\ref{sec: theory}. To ensure $\tilde\pi^*(\cdot)\approx \tilde\pi(\cdot)$ without the correction term $\whV$ in Equation~\eqref{equ: pi star and pi}, we will simply define $\whV^*$ to be $\whV^*-\whV$ in the following.

\paragraph*{Revisit the DTL example}
In the DTL example, all the statistics involved are averages of i.i.d. random variables, thus bootstrap consistency should hold. Specifically, $\whZ$ consists of the first-stage group means of the data, thus we have $\sqrt {n_1}(\whZ^* - \whZ )\convdis\N(0, I_K)$ as $n_1\goinf$.
Moreover, $\whV^*=e_{k^*}\whU_{k^*}$, where $\whU_{k^*}$ is the difference between the first-stage mean and global mean in the winner's group. So we also have $\frac{1}{\sqrt{\frac{1}{n_1} - \frac{1}{n_1+n_2} }} (\whU^*_{k^*} - \whU_{k^*}) \convdis\N(0,1) $ as $n_1,n_2\goinf$.

\subsection{Estimation algorithm}

Given the training data $\{(\tZ_i^*,\ell_i^*)\}_{i=1}^B$, our objective is to estimate the probability $\PP{\ell_i^*=1\mid \tZ_i^*}$. We consider a function class $\{f(\cdot;\xi),\xi\in\calP \}$ parameterized by $\xi\in\calP$, which is optimized by minimizing the empirical cross-entropy loss
\begin{align*}
  \xi^*:= \argmin_{\xi\in\calP} -\sum_{i=1}^B \ell_i^* \log f(\tZ_i^*; \xi) + (1-\ell_i^*) \log (1-f(\tZ_i^*; \xi)).
\end{align*}
The function $\hat\pi(\cdot):= f(\cdot;\xi^* )$ is then used to estimate $\tilde\pi^*(\cdot)$ and subsequently $\tilde\pi(\cdot)$.

If the true $\tilde\pi^*$ lies within the function class $\{f(\cdot;\xi),\xi\in\calP \}$, meaning $\tilde\pi^*(\cdot)=f(\cdot; \xi_0)$ for some $\xi_0\in\calP$, then $\xi^*$ will converge to $\xi_0$ at a rate of $O(B^{-1/2})$ under regularity conditions on $\tilde\pi^*$. In our experiment, we use a multilayer feedforward neural network as the chosen function class. Neural networks are considered universal function approximators, capable of approximating any continuous function on a compact set given a sufficient number of hidden units \citep{hornik1989multilayer}. However, the quality of the resulted function $f(\cdot;\xi^*)$ depends on various factors, such as the size and quality of the training data, network architecture, optimization algorithm, hyperparameters, and regularization techniques. Fine-tuning these factors is essential for achieving optimal performance. Techniques like splitting the data into training, validation, and test sets can be employed. The validation set is used for tuning while the test set is reserved for assessing the accuracy of the estimator. In addition, we provide an alternative to assessing the accuracy of $\hat\pi$ in Section~\ref{sec: assessing accuracy}.



\subsection{Computing p-values}
\label{sec: discrete exp family}

After obtaining the estimated selection probability $\hat\pi$, we approximate the (unnormalized) conditional density of $\hat\theta\mid\{\whM=M, \whW=W \}$ in Equation~\eqref{equ: cond density tilde pi} by
\[
\varphi(x; \theta, \Sigma)\cdot \hat \pi(\Gamma x + W).
\]
To facilitate fast evaluation of the CDF and inverse CDF of this distribution, we further approximate it with a discrete exponential family. Suppose first that $\theta$ is a univariate parameter and $\Sigma=\sigma^2$. Then we choose a grid $\{x_g\}_{g=1}^G$ and define the density supported on the grid as follows:
\begin{align}
\label{equ: discrete exp family}
\frac{1}{C} \sum_{g=1}^G \exp(\frac{\theta x_g}{\sigma^2} - \frac{x_g^2}{2\sigma^2} ) \cdot \hat{\pi}(\Gamma x_g+W) \cdot \indc{x=x_g},
\end{align}
where $C=\sum_{g=1}^G \exp(\frac{\theta x_g}{\sigma^2} - \frac{x_g^2}{2\sigma^2} ) \cdot \hat{\pi}(\Gamma x_g+W) $ is the normalizing constant.
This density maintains proportionality to $\varphi(x;\theta,\sigma^2)\cdot \hat{\pi}(\Gamma x+W)$ at the specified grid points $\{x_g\}_{g=1}^G$. Moreover, the discrete distribution \eqref{equ: discrete exp family} offers readily computable CDF and inverse CDF, facilitating the calculation of p-values and the construction of confidence intervals based on the conditional distribution of $\hat\theta$. 

If $\theta\in\R^s$ and $s>1$, to conduct inference for $\theta_j$, we need to eliminate the nuisance parameters. To achieve this, we further condition on $\hat\theta^\perp:=\hat\theta - \Sigma_{\cdot,j} \Sigma_{j,j}^{-1}  \theta_j$, and the conditional density of $\hat\theta_j\mid\{\whM=M, \whW=W, \hat\theta^\perp=\theta^\perp \}$ is proportional to
\[
\varphi(x; \theta_j, \Sigma_{j,j})\cdot\hat\pi(\Gamma(\Sigma_{\cdot,j} \Sigma_{j,j}^{-1}  x +\theta^\perp) + W ).
\]
To conduct inference for $\theta_j$ based on this univariate density, we apply the same discrete exponential family described above.

\subsection{Assessing the accuracy}
\label{sec: assessing accuracy}
Throughout our method, we have introduced several approximations, including the normal approximation of the pre-selection distribution of $\hat\theta$, the replacement of $\tilde\pi(\cdot)$ with $\hat{\pi}(\cdot)$, and the utilization of a discrete exponential family to approximate the conditional distribution. Here, we provide a practical way to check the reliability of these approximations.

Let $D^*$ be a bootstrap sample and assume $\hat\theta^*,\whZ^*,\whU^*,\whV^*,\whW^*$ all exhibit bootstrap consistency. Then the pre-selection distribution of $\hat\theta^*$ can be approximated by $\N(\hat\theta, \Sigma)$ where $\hat\theta$ is the statistic computed from the original data $D$. Analogously, the conditional density of $\hat\theta^*\mid \{\whM^*=M, \whW^*=W^*\}$ is approximately proportional to
\begin{align*}
  \varphi(x;\hat\theta,\Sigma )\cdot \hat \pi(\Gamma x + W^* ).
\end{align*}
Let $\widehat{H}^*(\cdot)$ denote the CDF corresponding to the above density. If this CDF is an accurate estimate of the exact CDF of the conditional distribution of $\hat\theta^*\mid \{\whM^*=M, \whW^*=W^*\}$, then $\whH^*(\hat\theta^*)\mid \{\whM^*=M \}$ should be approximately a pivot that is uniformly distributed on $[0,1]$. This observation provides a way to assess the accuracy of the quality of the approximations. Specifically, we can repetitively draw bootstrap samples $D^*$ such that the selection event $\{\whM^*=M\}$ happens, and compute the CDF $\widehat{H}^*(\hat\theta^*)$. These values are expected to be distributed approximately uniformly. If these computed values appear to be uniformly distributed, then it suggests that the estimated conditional distribution is a good approximation. We summarize this procedure in Algorithm~\ref{algo: check} in the case where $\hat\theta\in\R$.

\begin{algorithm}
\caption{Assessing the accuracy}
\label{algo: check}
\SetKwInOut{Input}{Input}
\SetKwInOut{Output}{Output}
\Input{Estimated selection probability function $\hat\pi$; target number $B$ of pivots to compute; observed value of $\hat\theta$; variance $\sigma^2$ of $\hat\theta$ }
\Output{Values of $\widehat{H}^*(\hat\theta^*)$ evaluated on the bootstrap datasets where the selection event occurs. }
$i=0$\\

\While {$i<B$}{
    Draw a dataset $D^*$ by bootstrapping\\
    Run the selection algorithm on $D^*$ to obtain $M^*$.\\
    \If {$M^*=M$}{
    $i\leftarrow i+1$\\
    Compute the statistics $\hat\theta^*$ and $W^*$\\
    Let $h^*(x)$ be the density proportional to
    $$h^*(x)\propto \varphi(x;\hat\theta, \sigma^2)\cdot \hat\pi(\Gamma x+W^*)$$\\
    Approximate the distribution with density $h^*(x)$ by the discrete exponential family introduced in Section~\ref{sec: discrete exp family}\\
    Evaluate the approximate CDF $\widehat{H}^*$ at $\hat\theta^*$
    }
}
  \end{algorithm}

\section{Theoretical analysis}
\label{sec: theory}

In the previous sections, we derived our method by assuming that $(\hat\theta, \whV, \whW)$ are jointly independently Gaussian. In this section, we will provide a guarantee that, under the asymptotic normality, the obtained confidence interval achieves the coverage probability for a single parameter of interest $\theta\in\R$. Let $\calF_n$ represent the class of distributions under consideration. Suppose that the dataset $D_n$ is drawn from the distribution $\bbF_n\in\calF_n$. Let $D_n^*$ denote a bootstrap sample generated from the original dataset $D_n$. We will use $\bbF_n^*$ to denote the bootstrap distribution conditional on $D_n$. Let $M_n$ be the selected model and let $\theta_n=\theta(\bbF_n,M_n)$ be the parameter of interest, which depends on the unconditional generating distribution of the data and the selected model. 

We assume that the statistic $\hat\theta_n$ satisfies $\sqrt n(\hat\theta_n - \theta_n)\convdis\N(0,\Sigma)$. Let $\whZ_n\in\R^d$ denote the basis computed from the original dataset $D_n$. Let $\Gamma_n=\Cov{\whZ_n, \hat\theta_n }\Var{\hat\theta_n}^{-1} $ and define $\whU_n=\whZ_n-\Gamma_n\hat\theta_n$. In our analysis, we assume all the covariance matrices are known. However, having uniformly consistent estimate of variance would suffice as in \cite{markovic2016bootstrap,tian2018selective}. Moreover, suppose $\whV_n\in\R^d$ is a statistic satisfying $\sqrt n \whV_n\convdis\N(0,\Sigma_V)$ and $\Cov{\whV_n, \hat\theta_n}=o_p(n^{-1})$. Note that $V_n$ can be trivially a constant 0. Lastly, let $\whW_n=\whU_n-\whV_n$ and $\tZ_n=\whZ_n - \whV_n=\Gamma_n\hat\theta_n + \whW_n $. For the bootstrap data $D_n^*$, we define the corresponding $\hat\theta^*_n,\whZ^*_n,\whU^*_n,\whV^*_n,\whW^*_n,\tZ^*_n$ similarly.

Our method involves estimating $\tilde\pi_n(\tZ)=\EE{\pi(\whZ_n)\mid\tZ_n=\tZ }$ using an estimator denoted as $\hat\pi_n$. Let $\whH_n(\cdot;\theta_n,\whW_n)$ be the CDF of the estimated conditional distribution of $\hat\theta_n$. Specifically, we define
\[
  \whH_n(x;\theta_n,\whW_n)=\frac{\int_{t\leq x} \varphi(t;\theta_n,\Sigma/n) \cdot \hat\pi(\Gamma_n t+\whW_n) \rd t}{\int_{\R} \varphi(t;\theta_n,\Sigma/n) \cdot \hat\pi(\Gamma_n t+\whW_n) \rd t },
\]
which leads to the construction of the interval
\begin{align}
  \calI(\hat\theta_n)=\calI(\hat\theta_n;\whW_n):=\left\{\theta_n: \alpha/2\leq \widehat H_n(\hat\theta_n;\theta_n,\whW_n)\leq 1-\alpha/2 \right\}.
  \label{equ: CI estimator}
\end{align}
We will show that under assumptions to be stated below, this interval achieves the asymptotic coverage guarantee in the sense that
\begin{align}
\label{equ: asymp coverage}
\liminf_{n\goinf}\inf_{\bbF_n\in\calF_n}\bbF_n\left\{\theta_n\in\calI(\hat\theta_n;\whW_n) \right\}\geq 1-\alpha.
\end{align}

First, we assume that the distribution of $R_n:=\sqrt n(\hat\theta_n - \theta_n, V_n)$ is close to a normal distribution in the sense that their Wasserstein-1 distance $W_1$ converges uniformly to 0. The $W_1$ distance between two distributions $\mu$ and $\nu$, using Kantorovich–Rubinstein duality \citep[Theorem 5.10]{villani2009optimal}, is defined as 
\[
W_1(\mu,\nu):=\sup\left\{\int f \rd \mu - \int f \rd \mu: f \text{ is }1 \text{ Lipschitz} \right\}.
\]
\begin{assumption}[Asymptotic normality of pre-selection distribution]\label{assump: pre-selection distribution}
Assume that
\begin{align*}
\lim_{n\goinf}\sup_{\bbF_n\in\calF_n}W_1(R_n, R_\infty )=0,\quad \text{ where } R_\infty\sim\N\left(\begin{pmatrix}\bzero \\ \bzero\end{pmatrix}, \begin{pmatrix}
\Sigma & \bzero \\ \bzero & \Sigma_V
\end{pmatrix}\right).
\end{align*}
\end{assumption}
Convergence in $W_1$ is stronger than weak convergence alone. In fact, convergence in the $W_1$ distance is equivalent to weak convergence combined with convergence of the first moment \citep[Theorem 5.11]{santambrogio2015optimal}, which is also a relatively mild condition. 
Next, we assume that the root $R_n^*=\sqrt n(\hat\theta_n^* - \hat\theta_n, \widehat V_n^*)$ from the bootstrap data converges in the same sense to $R_\infty$.
\begin{assumption}[Bootstrap consistency]\label{assump: bootstrap consistency}
Assume that
  \begin{align*}
    \lim_{n\goinf}\sup_{\bbF_n\in\calF_n} W_1(R_n^*, R_\infty )=0.
  \end{align*}
\end{assumption}
The weak convergence of $R_n^*$ to $R_\infty$ holds under various conditions. For example, if all the statistics are averages of $n$ samples, then the weak convergence follows from the bootstrap consistency of non-parametric means, see e.g. \cite[Theorem 15.4.5]{lehmann1986testing}. It would also be true if the statistics are \emph{linearizable statistics} \citep{chung2013exact}, which is a common assumption in the literature to establish asymptotic coverage guarantees \citep{tian2018selective,markovic2016bootstrap}.

\begin{assumption}[Smooth selection probability]\label{assump: smooth pi}
The selection probability function $\pi(Z)$ is Lipschitz continuous within a neighborhood of $\whZ_n$. Formally, we assume there exists $\delta_0>0$, such that $\pi(\cdot)$ is $L$-Lipschitz continuous within$\{Z\in\R^d:\|Z-\whZ_n\|_\infty\leq \delta_0 \}$.
\end{assumption}
Since $\pi(\cdot)$ is completely determined by the selection algorithm, this assumption implies that the selection algorithm remains relatively consistent and does not exhibit rapid and unpredictable changes due to minor variations in the input dataset. This is crucial because if $\pi(\cdot)$ were to exhibit abrupt fluctuations due to small perturbations in the input data, accurate estimation of the function $\pi(\cdot)$ would become elusive. The Lipschitz continuity property of $\pi(\cdot)$ also complements Assumptions~\ref{assump: pre-selection distribution} and \ref{assump: bootstrap consistency}, where the convergence of distributions is quantified using the Wasserstein 1 distance.

The next assumption concerns the algorithm used to estimate the selection probability function. Recall that the training data $\{(\tZ_i^*,\ell_i^*)\}_{1\leq i\leq B}$ described in Section~\ref{sec: generate training} satisfies $\PP{\ell_i^*=1}=\tilde\pi^*_n(\tZ_i^*)$. Hence, assuming an adequately representative function class and sufficient training data, the estimation of $\tilde\pi^*_n$ can be expected to be accurate. This assumption is formulated as follows.
\begin{assumption}[Estimate of $\tilde\pi^*$.]\label{assump: hat pi}
Assume that
\begin{align*}
\lim_{n\goinf}\sup_{\bbF\in\calF_n}\bbF_n\left\{\int_{\R} \varphi(t;\theta_n,\Sigma/n)  \big|\hat\pi_n(\Gamma_n t+\whW_n) - \tilde\pi^*_n(\Gamma_n t+\whW_n)\big| \rd t\geq\ep \right\}=0.
\end{align*}
\end{assumption}
This assumption implies that, after integrating out $\hat\theta_n$ over its approximated normal distribution, the difference between $\hat\pi_n$ and $\tilde\pi^*_n$ converges in probability to 0 uniformly. As $\varphi(t;\theta_n,\Sigma/n)$ concentrates around $\theta_n$, the assumption highlights the significance of accurately estimating $\tilde\pi^*_n(\cdot)$ particularly along the direction of $\Gamma_n$ and within a neighborhood of the observed $\tZ_n$.
We are now ready to state the central theorem which establishes that $\calI(\hat\theta_n;\whW_n)$ achieves an asymptotic coverage probability of $1-\alpha$.
\begin{theorem}[Asymptotic coverage probability]\label{thm}
Under Assumptions~\ref{assump: pre-selection distribution}, \ref{assump: bootstrap consistency}, \ref{assump: smooth pi} and \ref{assump: hat pi}, Equation~\eqref{equ: asymp coverage} holds true. That is, the interval $\calI(\hat\theta_n;\whW_n)$ in Equation~\eqref{equ: CI estimator} possesses an asymptotic coverage probability of $1-\alpha$.
\end{theorem}

\section{Applications and simulations}
\label{sec: simu}

In this section, we apply the proposed methodology to perform post-selection inference for various problems.

The implementation details of our blackbox method (BB) are as follows. We generate a training dataset consisting of 3000 data points. We include the original dataset in the training set so that there will be at least one positive label, i.e. selecting the model $M$. In situations where either the fraction of positive or negative labels is less than 10\%, we duplicate these data points. This replication ensures that both labels constitute at least 20\% of the data, creating a more balanced training set. The architecture of the neural network consists of three hidden layers, each comprising 200 neurons. The hidden layers use the ReLU activation function while the output layer uses the sigmoid activation. The optimization is performed using Adam \citep{kingma2014adam} with a learning rate 0.001. We employ a minibatch size of 200 and train the neural network for 3000 epochs. After obtaining an estimator $\hat\pi$ of $\tilde\pi$, we approximate the conditional distribution using the discrete exponential family given in Equation~\eqref{equ: discrete exp family}, where the 100 grid points $\{x_g\}_{g=1}^G$ are equally spaced on the interval $\hat\theta\pm 10 \hat\sigma$, where $\hat\sigma^2$ denotes the estimated variance of $\hat\theta$.

All the covariances including $\sigma^2,\Gamma$ are estimated using the same bootstrapped dataset employed for generating the training data. We aim for a target coverage probability of $1-\alpha=0.9$ for the confidence intervals. The evaluation is based on two metrics: the average coverage probability, denoting the proportion of confidence intervals that correctly cover the target parameters, and the average interval lengths. We repeat each experiment 200 times and report the average coverage probabilities and interval lengths. Error bars depicted in the figures represent 95\% confidence intervals constructed through bootstrapping from the 200 replicates.

\subsection{Drop the loser}

We begin by considering the DTL example described in Section~\ref{sec: dtl}.
The observations are generated from $K=50$ groups, with each group containing $n_1$ observations. The winner's group $k^*$ observes another $n_2=n_1/4$ observations in the second stage. We set all $\theta_k=0$, meaning that there is no effect in all groups. 
We define the test statistic to be $\hat\theta=\frac{n_1\widebar X_{k^*} + n_2 \widebar Y}{n_1+n_2}$ and let $\hat s^2$ denote the estimated variance of the data point. We compare the following four methods:
\begin{itemize}
\item Naive: Construct confidence intervals as $\hat\theta\pm \Phi^{-1}(1-\alpha/2) \hat s/\sqrt{n_1+n_2} $, disregarding the selection effect.

\item Splitting: Construct confidence intervals using solely the second-stage data as $\widebar Y\pm \Phi^{-1}(1-\alpha/2) \hat s/\sqrt{n_2}  $.

\item BB: the proposed method with the basis $\whZ=(\widebar X_1,\ldots,\widebar X_K )$ as introduced in Section~\ref{sec: dtl}

\item BB+marginalized: the proposed method using the basis $\tZ$ such that $\tZ_{k^*}=\hat\theta$ and $\tZ_{j}=\widebar X_j$ for $j\neq k^*$ as described in Section~\ref{sec: condition less}
\end{itemize}

The results are presented in Figure~\ref{fig: dtl}. The $x$-axes represent the sample size $n_1$, which is varied across $\{100,200,400\}$. The left panel of the figure demonstrates the average coverage probabilities and the right panel shows the average interval lengths.
A few observations are in order:
\begin{itemize}
\item The Naive method, which ignores the selection effect, exhibits significant under-coverage.

\item The Splitting method, which only uses the second-stage data for inference, achieves the desired coverage level but yields the longest intervals.

\item The two blackbox methods, BB and BB+marginalized, successfully achieve the intended coverage while producing shorter intervals in comparison to the splitting method. This finding illustrates the improved power of the conditional approach over the splitting method.

\item Furthermore, the intervals produced by the BB+marginalized method are shorter than those generated by the BB method. This observation aligns with our expectation that conditioning less leads to shorter intervals on average.
\end{itemize}

\begin{figure}
  \centering
  \includegraphics[width=.8\textwidth]{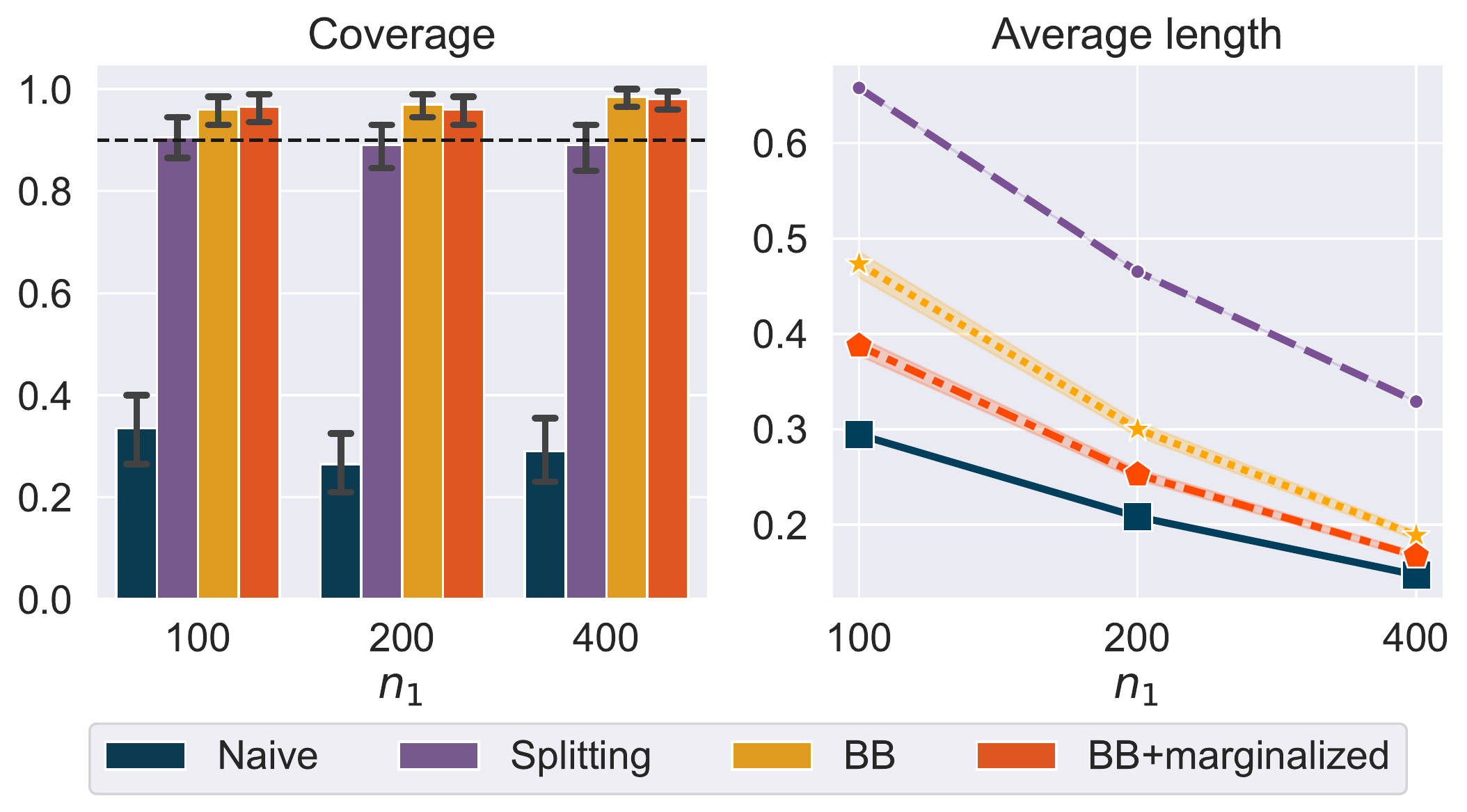}
  \caption{Average coverage probabilities (left panel) and interval lengths (right panel) for the drop-the-loser simulation. The intended coverage probability is 0.9, as depicted by the dashed line. The error bars represent the 95\% confidence intervals produced by bootstrapping from the 200 repetitions. The $x$-axes correspond to the first-stage sample size $n_1$. }
  \label{fig: dtl}
\end{figure}

\paragraph*{Assessing the accuracy}
In Section~\ref{sec: assessing accuracy}, we provided a method to evaluate the accuracy of the estimated selection probability. Here, we apply this procedure to a simulation to demonstrate its usage. 
Specifically, we apply Algorithm~\ref{algo: check} to generate 300 pivots under the bootstrap distribution and plot the empirical CDF of the pivots in Figure~\ref{fig: dtl empirical cdf}. Notably, the orange line, representing the empirical CDF of the generated pivots, closely aligns with the CDF of the uniform distribution (depicted by the dotted line). This alignment suggests that the approximated conditional distribution is indeed accurate. In contrast, the empirical CDF based on the unadjusted distribution, i.e. with $\hat\pi\equiv1$, significantly deviates from the uniform distribution, as indicated by the blue line.

\begin{figure}
\centering
\includegraphics[width=.5\textwidth]{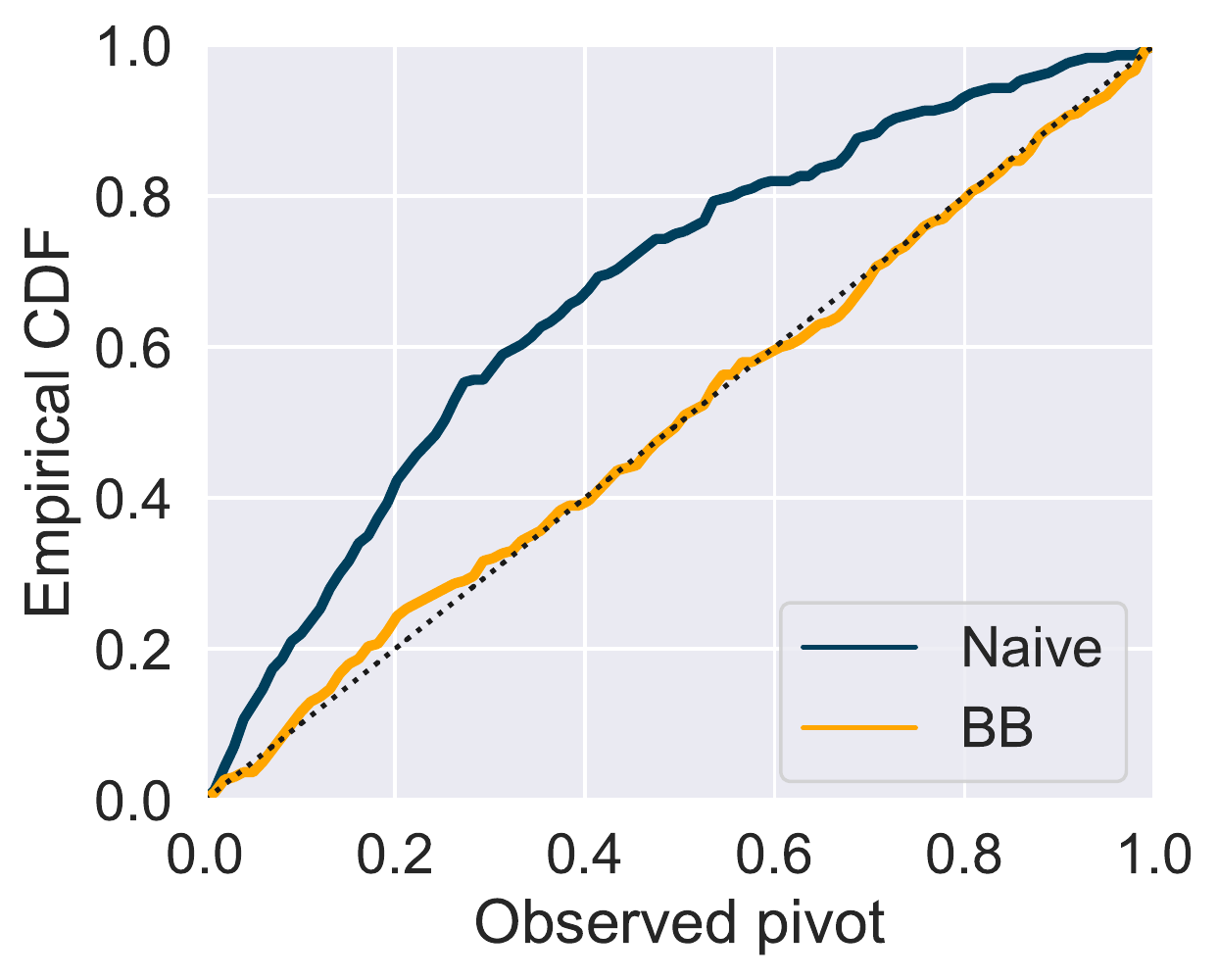}
\caption{Evaluating the accuracy of the estimated conditional distribution. The orange line represents the empirical CDF of the 300 pivots obtained by Algorithm~\ref{algo: check} using the estimated $\hat\pi$. The dotted line represents the CDF of the uniform distribution on $[0,1]$. The blue line represents the CDF of the pivots when no adjustment is applied, i.e. Algorithm~\ref{algo: check} is applied with $\hat\pi$ set to a constant 1. The alignment between the orange and the dashed lines indicates the accuracy of the estimated conditional distribution.}
\label{fig: dtl empirical cdf}
\end{figure}

\subsection{Lasso}
\label{sec: simu lasso}

Next, we proceed to apply the proposed method to what is arguably one of the most important post-selection inference problems: conducting inference after the lasso selection. As mentioned earlier, there have been several recent proposals for performing the randomized lasso to enhance the power of inference following selection. To explore this scenario, we employ the lasso with data carving, which uses 80\% of the data for the lasso selection. 

We consider a setup where the number of observations is set to $n=400$ and the number of features is $p=50$. The data is generated as follows: the observations $x_i$ are drawn independently from $\N_p(0,\Sigma_X)$ and $y_i\mid x_i\sim\N(x_i\tran\beta, 1)$. The covariance matrix $\Sigma_X$ is chosen to be the auto-regressive matrix with the $(i,j)$-entry being $0.3^{|i-j|}$. The regression coefficient vector $\beta$ is designed to be a sparse vector with 10 nonzero coefficients, which are equal to $\pm\sqrt{2c_0\frac{\log p}{n}}$ with random signs. Here, $c_0$ is the signal strength that will be varied across $\{0.6, 0.9, 1.2\}$. The lasso regularization parameter is fixed to be the constant $\sqrt{\log(p) / n}$ as suggested by \cite{negahban2012unified}. The above setup closely follows the simulation conducted in \cite{panigrahi2022approximate}.

For the proposed BB method, we choose the basis to be $\whZ=\frac1{n_1} X^{(1),\intercal} Y^{(1)}$, where $(X^{(1)},Y^{(1)})\in \R^{n_1\times (p+1)}$ is the random subset of data used for the lasso. We take $n_1=\lfloor 0.8n \rfloor$. The basis $\whZ$ is an average of $n_1$ i.i.d. quantities, thus is expected to be approximately normally distributed. More importantly, the lasso selection is completely characterized by the quantities $X^{(1),\intercal} X^{(1)}$ and $X^{(1),\intercal} Y^{(1)}$. Given that the design matrix has been normalized such that $X^{(1),\intercal} X^{(1)}$ is essentially constant, it is not necessary to include it in the basis. 

We consider the following competing methods. The Naive method constructs the classic Wald-type confidence intervals ignoring the selection effect. The Splitting method uses solely the hold-out 20\% data for inference. Moreover, we consider the selective MLE method recently proposed by \cite{panigrahi2022approximate}. This MLE method constructs Wald-type confidence intervals based on the MLE of the selection-adjusted likelihood and the corresponding Fisher information matrix. We use the implementation available at the GitHub repository\footnote{\url{https://github.com/jonathan-taylor/selective-inference}}.

The results are presented in Figure~\ref{fig: random lasso}. The $x$-axes correspond to the signal strength $c_0$. A few observations are in order:
\begin{itemize}
\item As anticipated, the Naive intervals do not achieve the correct coverage. The Splitting intervals achieve the intended coverage, but at the expense of long interval lengths.

\item Both the proposed BB method and the MLE method from \cite{panigrahi2022approximate} achieve the desired coverage and exhibit similar interval lengths.  

\end{itemize}

It is worth emphasizing that the MLE method is tailored specifically for the lasso problem, whereas our proposed approach is a generic algorithm that does not rely on specific structures unique to the lasso. This highlights the potential applicability of the proposed method in scenarios where no specialized algorithm exist or is possible. We will explore these scenarios in the subsequent examples.

\begin{figure}
\centering
\includegraphics[width=.8\textwidth]{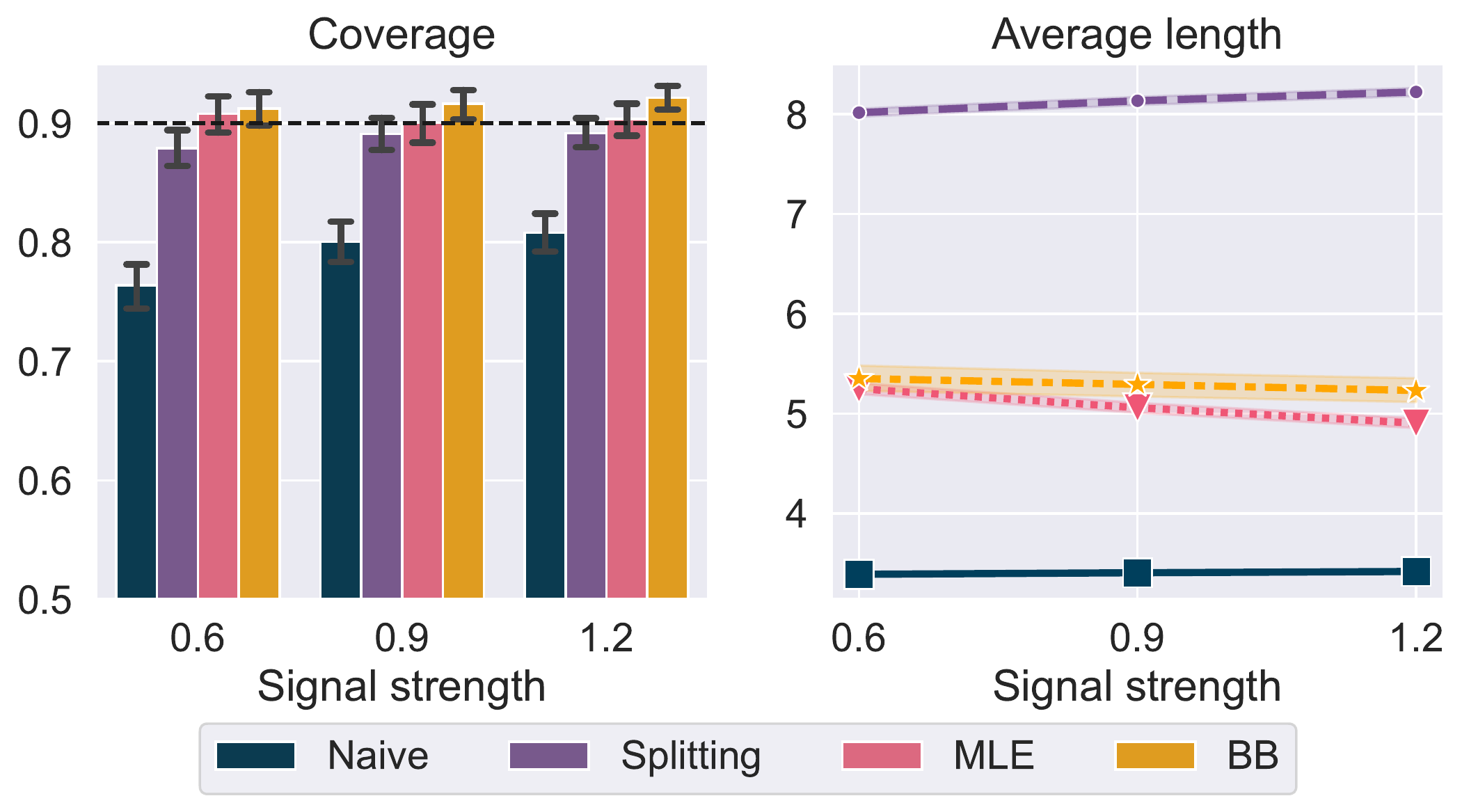}
\caption{Average coverage probabilities and lengths for the lasso simulation. The $x$-axes correspond to the signal strength $c_0$. The protocol is similar as in Figure~\ref{fig: dtl}.
}
\label{fig: random lasso}
\end{figure}

\subsection{Knockoff}

The knockoff filter \citep{barber2015controlling} offers a methodology for selecting variables in a linear regression model while controlling the false discovery rate (FDR). It is not itself a variable selection algorithm, but operates on existing ones, such as the lasso, to ensure FDR control. Our goal here is to construct valid confidence intervals for the variables selected by the knockoff filter. 

We use the same data generated in the lasso example in Section~\ref{sec: simu lasso}. We perform variable selection using the Gaussian Model-X knockoff with the lasso algorithm, targeting an FDR at 0.2. The basis $\whZ$ is chosen to be the same as in the lasso example.

The results are presented in Figure~\ref{fig: knockoff}. It is evident that the Naive method, which ignores the selection effect, fails to achieve the desired coverage. However, our proposed method successfully achieve the intended coverage probability. Consequently, our approach provides a solution for conducting valid post-selection inference in scenarios involving more complex selection procedures that were previously deemed infeasible.

\begin{figure}
  \centering
  \includegraphics[width=.8\textwidth]{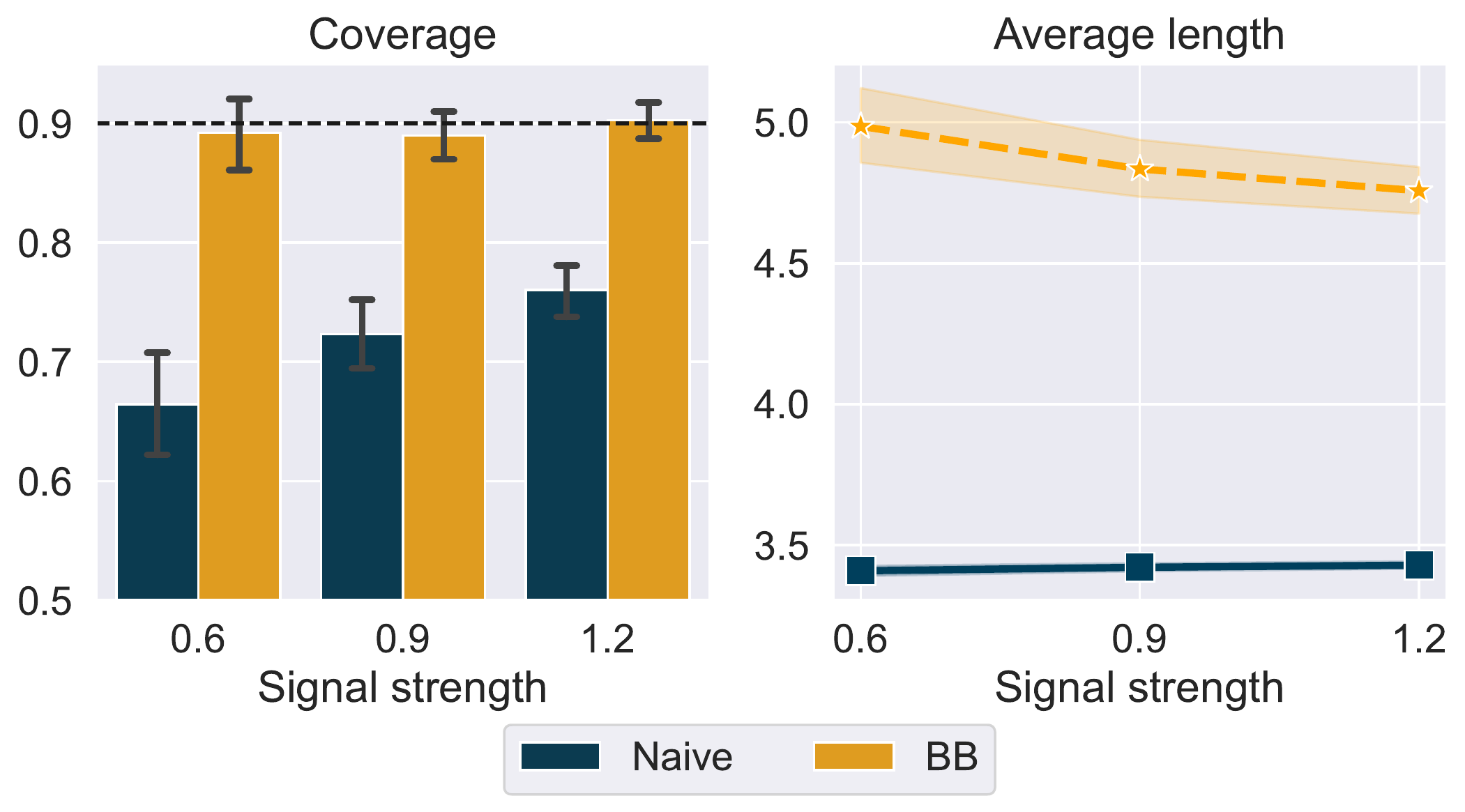}
  \caption{Average coverage probabilities and lengths for the knockoff simulation. The $x$-axes correspond to the signal strength $c_0$. The protocol is similar as in Figure~\ref{fig: dtl}.}
  \label{fig: knockoff}
\end{figure}




\subsection{Benjamini-Hochberg procedure}

The Benjamini-Hochberg (BH, \citep{benjamini1995controlling}) procedure is a multiple testing procedure that controls the false discovery rate (FDR). In our context, we apply the BH procedure to identify a subset of potentially non-null effects and subsequently perform inference for the selected effects.

Consider a scenario with $K=20$ treatments denoted by $\theta_k$ ($1\leq k\leq K$). For each treatment group, we gather $n=300$ independent observations $x_{k,i}\sim\N(\theta_k, 1)$ and compute the mean effect $\widebar X_k=\frac1n\sum_{i=1}^n x_{k,i}$. We then compute the p-values as $p_k=2\Phi(-\sqrt n |\widebar X_k|)$ for $1\leq k\leq K$, and apply the BH procedure on $p_1,\ldots,p_K$ with the target FDR set to 0.2. The parameters are set to be $\theta_{1:4}=-\theta_{5:8}= \theta_0$, and $\theta_{9:20}=0$. We vary the signal strength $\theta_0$ across $\{0.05, 0.1, 0.2\}$. For our blackbox method, the basis $\whZ$ is chosen to consist of the group means $(\widebar X_1,\ldots,\widebar X_K)$, as they completely determine the selection of the parameters.

The results are presented in Figure~\ref{fig: bh}. Consistent with previous examples, the proposed BB method effectively adjusts for the selection effect introduced by the BH procedure. Without the adjustment, the Naive method is severely biased especially in the situation with a weak signal strength.

\begin{figure}
\centering
\includegraphics[width=.8\textwidth]{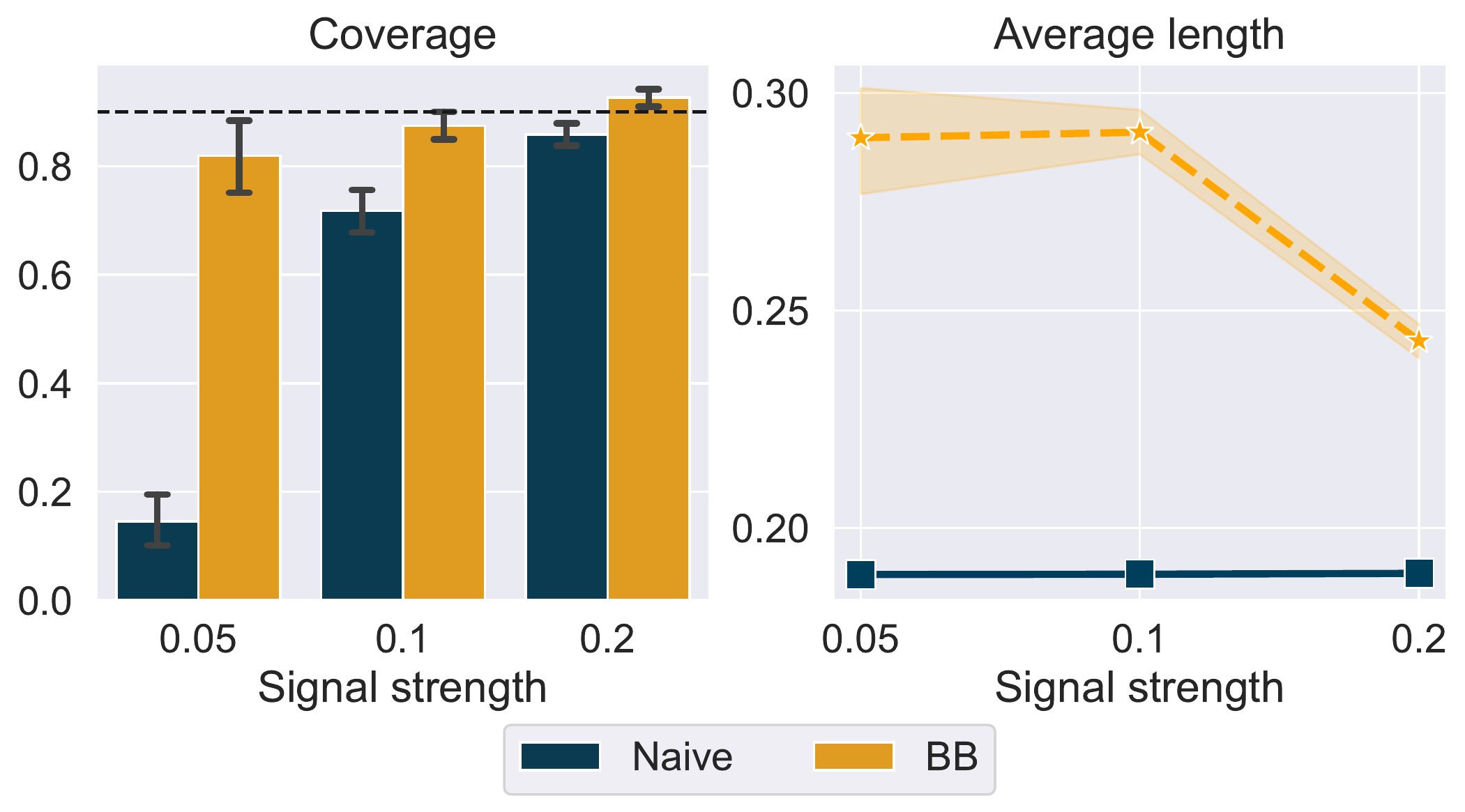}
\caption{Average coverage probabilities and lengths for the BH example. The $x$-axes correspond to the signal strength $\theta_0$. The protocol is similar as in Figure~\ref{fig: dtl}.}
\label{fig: bh}
\end{figure}

\subsection{Repeated significance test}

Repeating a hypothesis test while accumulating more data until achieving significance can lead to an increased risk of Type I errors \citep{armitage1969repeated}. In this context, we examine a scenario where data is iteratively gathered until a two-sample $t$-test achieves significance. This situation is particularly relevant in commercial A/B testing, where experimenters can continuously monitor the sample size while looking at the reported $p$-values.

Consider two populations $\N(\mu_1,\sigma^2)$ and $\N(\mu_2,\sigma^2)$. Suppose one is interested in whether $\mu_1=\mu_2$. We start with 100 observations drawn from both populations and perform a two-sample $t$-test to test the hypothesis $\mu_1=\mu_2$, which is rejected if the p-value is smaller than 0.1. If it is not rejected, an additional $50$ observations are sampled from both populations and the $t$-test is repeated again using the combined data. This process is repeated until the $t$-test is significant. Our goal here is to construct a confidence interval for the effect difference $\mu_1-\mu_2$ utilizing the data collected up to the point when a significant outcome is achieved in the $t$-test.

Suppose that the two-sample $t$-test is rejected at the $T$-th stage. In our proposed method, to generate one pair of training data $(\tZ_i^*, \ell_i^*)$, we bootstrap the same amount of data in $T$ stages in the same manner as described above. A two-sample $t$-test is conducted using the accumulated bootstrap data at each of the $T$ stages. If the none of the $T$ tests is rejected, we let $\ell_i^*=0$; otherwise, $\ell_i^*=1$. The basis $\tZ_i^*$ consists of the sample means and sample standard deviations at each of the $T$ stages for the two samples of data. This is because the two-sample $t$-tests are completely determined by the sample means and standard deviations. 

The results are shown in Figure~\ref{fig: repeated}.  The $x$-axes represents the effect size $\mu_1-\mu_2$. The naive method constructs confidence intervals using all the accumulated data without adjusting for the selection effect. When the effect size is 0, we observe that the naive intervals hardly ever cover the true parameter. As the effect size increases, the selection effect becomes weaker, thus the naive interval has higher coverage. In contrast, the proposed method consistently achieves the advertised coverage probability across all scenarios.

\begin{figure}
  \centering
  \includegraphics[width=.8\textwidth]{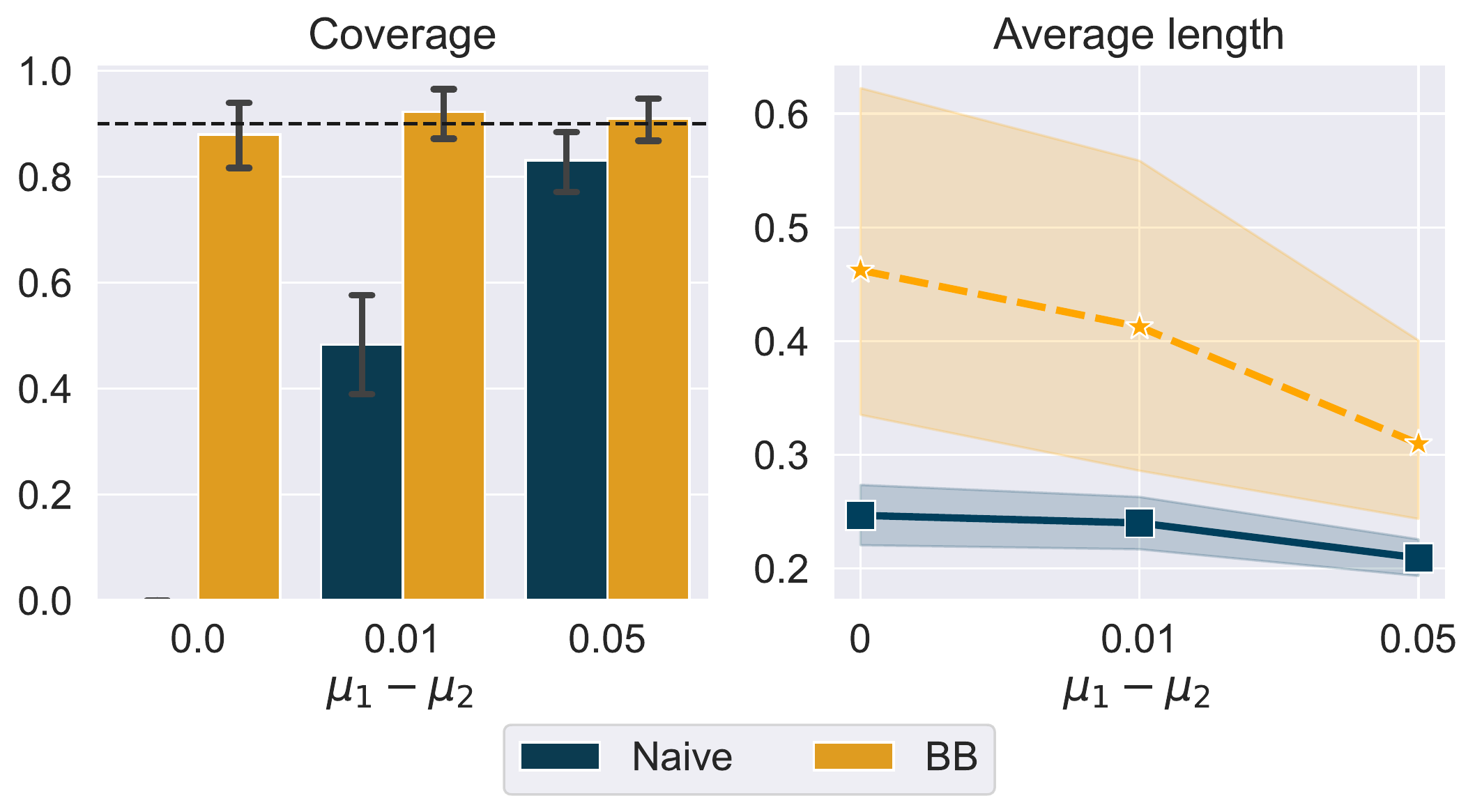}
  \caption{Average coverage probabilities and lengths for the parameter $\mu_1-\mu_2$ in the repeated significantly testing example. The $x$-axes correspond to the effect size $\mu_1-\mu_2$. The protocol is similar as in Figure~\ref{fig: dtl}}
  \label{fig: repeated}
\end{figure}

\section{Conclusion}

This paper introduced a versatile approach for conducting conditional selective inference following a selection procedure where an exact characterization of the selection event is not readily available. Our method involves repeatedly executing the selection algorithm on bootstrapped datasets to gather information about the selection probability. Despite its computational demands, this approach provides a solution that extends the scope of conditional selection inference, previously constrained to simple selection rules like the lasso. We demonstrated the usage and effectiveness of the proposed approach through a series of applications.

\bibliographystyle{apalike}
\bibliography{arxiv.bbl}

\clearpage

\appendix
\section{Proofs}
\label{sec: appendix}

\subsection{Proof of Theorem \ref{thm}}
\label{prf: thm}
\begin{proof}
By Lemma \ref{lemma}, it suffices to prove that for any $\ep>0$, 
\begin{align*}
\lim_{n\goinf}\sup_{\bbF_n\in\calF_n}\bbF_n\left\{\sup_{x\in\R^s} |\widehat H_n(x; \theta_n,\whW_n) - H_n(x; \theta_n,\whW_n)|\geq\ep \right\}=0.
\end{align*}
Moreover, it suffices to prove
\begin{align}
  \label{equ: prf uniform conv of CDF}
  \lim_{n\goinf}\sup_{\bbF_n\in\calF_n}\bbF_n\left\{\sup_{x\in\R^s} \bigg|\int_{t\leq x} p_{\hat\theta_n,\theta_n}(t)\cdot \tilde\pi_n(\Gamma_n t+\whW_n )\rd t - \int_{t\leq x} \varphi(t;\theta_n,\Sigma/n)\cdot \hat{\pi}_n(\Gamma_n t+\whW_n)\rd t \bigg|\geq\ep \right\}=0,
\end{align}
where $p_{\hat\theta_n,\theta_n}$ represents the density of $\hat\theta_n$ parametrized by $\theta_n$. We decompose the error in the above display into two terms by triangular inequality:
\begin{align*}
&|\int_{t\leq x} p_{\hat\theta_n,\theta_n}(t)\cdot \tilde\pi_n(\Gamma_n t + \whW_n)\rd t - \int_{t\leq x} \varphi(t;\theta_n,\Sigma/ n)\cdot \hat\pi_n(\Gamma_n t+\whW_n)\rd t|\\
&\leq |\int_{t\leq x} p_{\hat\theta_n,\theta_n}(t)\cdot \tilde\pi_n(\Gamma_n t + \whW_n)\rd t - \int_{t\leq x}  \varphi(t;\theta_n,\Sigma/n)\cdot \tilde\pi^*_n(\Gamma_n t + \whW_n)\rd t| \\
&\quad +|\int_{t\leq x}  \varphi(t;\theta_n,\Sigma/n)\cdot \tilde\pi^*_n(\Gamma_n t + \whW_n)\rd t - \int_{t\leq x} \varphi(t;\theta_n,\Sigma/n)\cdot \hat\pi_n(\Gamma_n t+\whW_n)\rd t|
\end{align*}

To bound the first term, we use the following lemma, which is proved in Section~\ref{prf: first term}
\begin{lemma}[Convergence of the first term]
  \label{lemma: first term}
  \begin{align*}
    \lim_{n\goinf} \sup_{\bbF_n\in\calF_n} \bbF_n \left\{\sup_{x\in\R^s} \bigg|\int_{t\leq x} p_{\hat\theta_n,\theta_n}(t)\cdot \tilde\pi_n(\Gamma_n t + \whW_n)\rd t - \int_{t\leq x}  \varphi(t;\theta_n,\Sigma/n)\cdot \tilde\pi^*_n(\Gamma_n t + \whW_n)\rd t\bigg|\geq \ep \right\} =0.
  \end{align*}
\end{lemma}

The second term satisfies
\begin{align*}
\lim_{n\goinf}\sup_{\bbF_n\in\calF_n}\bbF_n\left\{\sup_{x\in\R^s}|\int_{t\leq x}  \varphi(t;\theta_n,\Sigma/n)\cdot \tilde\pi^*_n(\Gamma_n t + \whW_n)\rd t - \int_{t\leq x} \varphi(t;\theta_n,\Sigma/n)\cdot \hat\pi_n(\Gamma_n t+\whW_n)\rd t|
\geq \ep\right\} =0
\end{align*}
by Assumption~\ref{assump: hat pi}. This proves Equation~\eqref{equ: prf uniform conv of CDF} and hence the theorem.

\end{proof}

\subsection{Technical lemma}
\begin{lemma}[Adapted from Lemma A.1 of \cite{romano2012uniform}]
  \label{lemma}
  Suppose $\hat\theta_n\sim H_n$ under the distribution $\bbF_n$ and $\whH_n$ is some estimator of $H_n$.
  If for any $\ep>0$, we have 
  \begin{align*}
  &\lim_{n\goinf}\sup_{\bbF_n\in\calF_n}\bbF_n\left\{\sup_{t\in\R}|\whH_n(t;\theta_n) - H_n(t;\theta_n) |\geq \ep\right\}=0 ,
  \end{align*}
  then 
  \begin{align*}
  \liminf_{n\goinf}\inf_{\bbF_n\in\calF_n}\bbF_n\left\{\whH_n^{-1}(\alpha/2)\leq \hat\theta_n\leq \whH_n^{-1}(1-\alpha/2) \right\}\geq1-\alpha.
  \end{align*}
  \end{lemma}
  \begin{proof}[Proof of Lemma \ref{lemma}]
  Note that
  \begin{align*}
  \bbF_n\left\{\hat\theta_n\leq \whH_n^{-1}(1-\alpha/2)\right\}&\geq \bbF_n\left\{\hat\theta_n\leq \whH_n^{-1}(1-\alpha/2)\text{ and }\sup_{t\in\R}|H_n(t)-\whH_n(t) |\leq\ep\right\}\\
  &\geq \bbF_n\left\{\hat\theta_n\leq H_n^{-1}(1-\alpha/2-\ep)\text{ and }\sup_{t\in\R}|H_n(t)-\whH_n(t) |\leq\ep\right\}\\
  &\geq 1-\alpha/2-\ep-\bbF_n\left\{\sup_{t\in\R}|H_n(t)-\whH_n(t) |\geq\ep \right\}.
  \end{align*}
  Taking infimum over $\bbF_n\in\calF_n$ on both sides and sending $n$ to $\infty$, we get
  \begin{align*}
  \liminf_{n\goinf}\inf_{\bbF_n\in\calF_n}\bbF_n\left\{\hat\theta_n\leq \whH_n^{-1}(1-\alpha/2)\right\}\geq 1-\alpha/2-\ep.
  \end{align*}
  Because this holds for any $\ep>0$, 
  \begin{align*}
  \liminf_{n\goinf}\inf_{\bbF_n\in\calF_n}\bbF_n\left\{\hat\theta_n\leq \whH_n^{-1}(1-\alpha/2)\right\}\geq 1-\alpha/2.
  \end{align*}
  By a similar argument, we have
  \begin{align*}
  \liminf_{n\goinf}\inf_{\bbF_n\in\calF_n}\bbF_n\left\{\hat\theta_n\geq \whH_n^{-1}(\alpha/2)\right\}\geq 1-\alpha/2.
  \end{align*}
  This concludes the proof.
  \end{proof}

\subsection{Proof of Lemma~\ref{lemma: first term}}
\label{prf: first term}
\begin{proof}
  Note that
  \begin{align*}
    J_n(x)&:=\int_{t\leq x} p_{\hat\theta_n,\theta_n}(t)\cdot \tilde\pi_n(\Gamma_n t + \whW_n)\rd t\\
    &=\int_{t\leq x} \int_{\R^d} \pi(\Gamma_n t+v+\whW_n)  p_{\hat\theta_n,\theta_n}(t) p_{\hat V|\tZ=\Gamma_n t+\whW_n }(v) \rd v\rd t\\
    &=\int_{\R^d\times \R}p_{\hat\theta_n, \hat V_n}(t,v) \pi(\Gamma_n t+v+\whW_n)\Indc{t\leq x} \rd v\rd t\\
    &=\int_{\R^d\times \R}p_{\sqrt n(\hat\theta_n-\theta_n), \sqrt n\hat V_n}(t,v) \pi(\Gamma_n (n^{-1/2}t+\theta_n )+n^{-1/2}v+\whW_n) \Indc{n^{-1/2}t+\theta_n\leq x} \rd v\rd t\\
    &=\EE[R_n]{\pi(\Gamma_n (n^{-1/2}t+\theta_n )+n^{-1/2}v+\whW_n) \Indc{n^{-1/2}t+\theta_n\leq x} }.
  \end{align*}
  Define
  \begin{align*}
    \underline{J}_n(x)&:=\int_{t\leq x} \varphi(t;\theta_n, \Sigma/n)\cdot \varphi(v;0,\Sigma_V/n)\cdot \pi(\Gamma_n t+v+\whW_n)\rd v\rd t\\
    &=\int_{\R^d\times \R} \varphi(t;0,\Sigma)\varphi(v;0,\Sigma_V) \pi(\Gamma_n (n^{-1/2} t+\theta_n) )+n^{-1/2}v+\whW_n)\Indc{n^{-1/2}t+\theta_n\leq x } \rd v\rd t\\
    &=\EE[R_\infty]{\pi(\Gamma_n (n^{-1/2}t+\theta_n )+n^{-1/2}v+\whW_n) \Indc{n^{-1/2}t+\theta_n\leq x}}.
  \end{align*}
  By Assumption~\ref{assump: smooth pi}, $\pi$ is Lipschitz inside a neighborhood of $\whZ_n$. By definition, $\whZ_n=\Gamma_n\hat\theta_n+\whV_n+\whW_n$. Since $\hat\theta_n=\theta_n+o_p(1)$, $\whV_n=o_p(1)$, $\pi$ is $L$-Lipschitz inside the region $\{ Z\in\R^d: \|Z-(\Gamma_n\theta_n+\whW_n)\|_\infty\leq\delta_0/2\}$ with probability going to 1. 
  Let $\calA_n=\{u\in\R^{d+1}: \|u\|_\infty\leq\sqrt n\delta_0/2 \}$. Then $\EE[R_n]{\Indc{\calA_n^c}}\to0$ and $\EE[R_\infty]{\calA_n^c}\to0$. Since $\pi$ is a bounded function, the integral of $\pi$ over the region $\calA_n^c$ goes to 0 under $R_n$ and $R_\infty$. So our focus is on $(t,v)\in \calA_n$, where $\pi(\Gamma_n (n^{-1/2}t + \theta_n ) +n^{-1/2} v + \whW_n  )$ is Lipschitz in $(t,v)$ with Lipschitz constant $n^{-1/2}L(1 + \|\Gamma_n\|_2)$.
  Let $x'=\min\{x, \sqrt n\delta_0 \}$.
  Fix some $\delta$ and define 
  \begin{align*}
    \xi(t)=\begin{cases}
      1 & t\leq x'\\
      \frac{1}{\delta}(x'+\delta - t ) & x'\leq t\leq x'+\delta\\
      0 & t\geq x'+\delta
    \end{cases}
  \end{align*}
  So $\xi$ is $1/\delta$-Lipschitz continuous and provides an upper bound of the indicator: $1\geq \xi(t)\geq \Indc{t\leq x'}$. Since $\pi$ is $n^{-1/2}L(1+\|\Gamma_n\|_2)$-Lipschitz continuous when $(t,v)\in\calA_n$ and $0\leq \pi\leq 1$, we have $\pi(\Gamma_n (n^{-1/2}t+\theta_n)+n^{-1/2}v+\whW_n)\cdot \xi(n^{-1/2}t+\theta_n)$ is $(n^{-1/2}L(1+\|\Gamma_n\|_2)+n^{-1/2}/\delta)$-Lipschitz continuous in $(t,v)$. Let $\bar L=n^{-1/2} L(\|\Gamma_n\|_2+1)+n^{-1/2}/\delta$. 
  Then by the definition of Wasserstein 1 distance and by the Lipschitzness of $\pi_{\calA}:=\pi\cdot\Indc{\calA_n}$, we have
  \begin{align*}
    J_n(x) - \underline{J}_n(x)&\leq \EE[R_n]{\pi_{\calA}(\Gamma_n (n^{-1/2}t+\theta_n )+n^{-1/2}v+W_n) \xi(n^{-1/2}t+\theta_n)}\\
    &-\EE[R_\infty]{\pi_{\calA}(\Gamma_n (n^{-1/2}t+\theta_n )+n^{-1/2}v+W_n) \xi(n^{-1/2}t+\theta_n)}\\
    &+\EE[R_\infty]{\pi_{\calA}(\Gamma_n (n^{-1/2}t+\theta_n )+n^{-1/2}v+W_n) \xi(n^{-1/2}t+\theta_n)}\\
    &-\EE[R_\infty]{\pi_{\calA}(\Gamma_n (n^{-1/2}t+\theta_n )+n^{-1/2}v+W_n) \Indc{n^{-1/2}t+\theta_n\leq x}}+o_p(1)\\
    &\leq \bar L \cdot W_1(R_n,R_\infty ) + \int_{\R^s} \Indc{x\leq n^{-1/2}t+\theta_n\leq x+\delta } \varphi(t;0,\Sigma)\rd t+o_p(1)\\
    &\leq \bar L\cdot W_1(R_n,R_\infty) + n^{1/2}\delta/ \sqrt{2\pi|\Sigma|}+o_p(1) \\
    &=n^{-1/2}L(\|\Gamma_n\|_2+1) W_1(R_n,R_\infty) + n^{-1/2}W_1(R_n,R_\infty)/\delta + n^{1/2}\delta/\sqrt{2\pi|\Sigma|}  + o_p(1).
  \end{align*}
  Choose $\delta=W_1(R_n,R_\infty)^{1/2} n^{-1/2}$. Then the above display is upper bounded by
  \begin{align*}
    n^{-1/2}L(\|\Gamma_n\|_2+1)W_1(R_n,R_\infty) + (1+1/\sqrt{2\pi|\Sigma|}) W_1(R_n,R_\infty)^{1/2} +o_p(1).
  \end{align*}
  For the same reason we can prove the other direction of the inequality and obtain
  \begin{align*}
    |J_n(x)-\underline{J}_n(x)|\leq n^{-1/2}L(\|\Gamma_n\|_2+1)W_1(R_n,R_\infty) + (1+1/\sqrt{2\pi|\Sigma|})W_1(R_n,R_\infty)^{1/2} +o_p(1).
  \end{align*}
  By assumption that $\lim_{n\goinf}\sup_{\bbF_n\in\calF_n} W_1(R_n,R_\infty)=0$, we have shown that
  \begin{align*}
    \lim_{n\goinf}\sup_{\bbF_n\in\calF_n}\bbF_n\{\sup_{x\in\R}|J_n(x)-\underline{J}_n(x)|\geq\ep \}=0.
  \end{align*}
  
  By the same argument using the assumption on the convergence of $W_1(R_n^*,R_\infty)$ for the bootstrap distribution, we have
  \begin{align*}
    \lim_{n\goinf}\sup_{\bbF_n\in\calF_n}\bbF_n\{\sup_{x\in\R}|J^*_n(x)-\underline{J}_n(x)|\geq\ep \}=0,
  \end{align*}
  where $J_n^*$ is similarly defined as $J_n$ with $\tilde\pi_n$ replaced by $\tilde\pi_n^*$. This concludes the proof of the lemma.
    
\end{proof}

\end{document}